\documentclass[a4paper]{article}

\usepackage{amsmath}
\usepackage{amssymb}
\usepackage{graphicx}
\usepackage[OT4]{fontenc}
\usepackage[latin2]{inputenc}
\usepackage{tikz}
\usepackage{url}
\usepackage[numbers]{natbib}
\usepackage{bm}
\usepackage{setspace}
\newcommand{\abs}[1]{\left| #1 \right|}

\newcommand{\okra}[1]{\left( #1 \right)}
\newcommand{\kwad}[1]{\left[ #1 \right]}
\newcommand{\klam}[1]{\left\{ #1 \right\}}

\newcommand{\floor}[1]{\left\lfloor #1 \right\rfloor}

\DeclareMathOperator{\essinf}{ess\, inf}
\DeclareMathOperator{\esssup}{ess\, sup}
\DeclareMathOperator{\card}{card}

\DeclareMathOperator{\var}{Var}
\DeclareMathOperator{\sred}{\mathbf{E}}
\newcommand{\peq}{\stackrel{+}{=}}
\newcommand{\pge}{\stackrel{+}{>}}
\newcommand{\ple}{\stackrel{+}{<}}

\newtheorem{definition}{Definition}
\newtheorem{corollary}{Corollary}

\newtheorem{theorem}{Theorem}
\newtheorem{lemma}{Lemma}
\newenvironment*{proof}{\begin{trivlist}\item[]
\noindent\textbf{Proof:}}{$\Box$\par\end{trivlist}}
\newenvironment*{proof*}[1]{\begin{trivlist}\item[]
\noindent\textbf{Proof of #1:}}{$\Box$\par\end{trivlist}}

\author{{\L}ukasz D\k{e}bowski\thanks{%Manuscript received ... ; revised ... .
    {\L}. D\k{e}bowski is with
    the Institute of Computer Science, Polish Academy of Sciences, 
    ul. Jana Kazimierza 5, 01-248 Warszawa, Poland 
    (e-mail: ldebowsk@ipipan.waw.pl).}}

\title{Hilberg Exponents: \\ New Measures of Long Memory in the Process} \date{}

\begin{document}

\pagestyle{empty}   
\begin{titlepage}
%\titlepage
\maketitle

\begin{abstract}
  The paper concerns the rates of power-law growth of mutual
  information computed for a stationary measure or for a universal
  code. The rates are called Hilberg exponents and four such
  quantities are defined for each measure and each code: two random
  exponents and two expected exponents. A particularly interesting
  case arises for conditional algorithmic mutual information. In this
  case, the random Hilberg exponents are almost surely constant on
  ergodic sources and are bounded by the expected Hilberg
  exponents. This property is a ``second-order'' analogue of the
  Shannon-McMillan-Breiman theorem, proved without invoking the
  ergodic theorem.  It carries over to Hilberg exponents for the
  underlying probability measure via Shannon-Fano coding and Barron
  inequality. Moreover, the expected Hilberg exponents can be linked
  for different universal codes. Namely, if one code dominates
  another, the expected Hilberg exponents are greater for the former
  than for the latter. The paper is concluded by an evaluation of
  Hilberg exponents for certain sources such as the mixture Bernoulli
  process and the Santa Fe processes.
% \end{abstract}
%
% \begin{keywords}
  \\[1em]
\end{abstract}

%MSC 2000: 
% 94A29 - Source coding
% 60G10 - Stationary processes
% 94A17 - Measures of information, entropy

\end{titlepage}
\pagestyle{plain}   

% \tableofcontents

\section{Preliminaries and main results}
\label{secIntroduction}

According to a conjecture by Hilberg
\cite{Hilberg90,CrutchfieldFeldman03}, the mutual information between
two adjacent long blocks of text in natural language grows like a
power of the block length. This property strongly differentiates
natural language from $k$-parameter sources, for which the mutual
information is proportional to the logarithm of the block length
\cite{Atteson99,BarronRissanenYu98,LiYu00}. In
\cite{Debowski11b,Debowski12} a class of stationary processes, called
Santa Fe processes, has been exhibited, which feature the power-law
growth of mutual information. Moreover, it was shown in
\cite{Debowski11b} that Hilberg's conjecture implies Herdan's law, an
integrated version of the famous Zipf's law in linguistics
\cite{Zipf35}.  Later, Dębowski \cite{Debowski13,Debowski13d} tested
Hilberg's conjecture experimentally by approximating the mutual
information with the Lempel-Ziv code \cite{ZivLempel77} and a newly
introduced universal code called switch distribution
\cite{Debowski13d}. Whereas the estimates of mutual information for
the Lempel-Ziv code grow roughly as a power law for both a
$k$-parameter source and natural language, cf.\
\cite{LouchardSzpankowski97}, the other code does reveal the
difference predicted by Hilberg's conjecture: the estimates of mutual
information for the switch distribution grow as a power law for
natural language whereas only logarithmically for a $k$-parameter
source \cite{Debowski14c}.

To provide more theory for Hilberg's conjecture, in this paper we
abstract from its empirical verification and we investigate the
bounding rates for power-law growth of mutual information evaluated
for an arbitrary stationary probability measure or for a universal
code. We call these rates Hilberg exponents, to commemorate Hilberg's
insight. The formal definition rests on the following preliminaries:
\begin{enumerate}
\item Let $\mathbb{X}$ be a countable alphabet.  Consider a
  probability space $(\Omega,\mathcal{J},Q)$ with
  $\Omega=\mathbb{X}^{\mathbb{Z}}$, discrete random variables
  $X_k:\Omega\ni(x_i)_{i\in\mathbb{Z}}\mapsto x_k\in\mathbb{X}$, and a
  probability measure $Q$ which is stationary on
  $(X_i)_{i\in\mathbb{Z}}$ but not necessarily ergodic.  Blocks of
  symbols or variables are denoted as $X_n^m=(X_i)_{n\le i\le m}$.  We
  introduce shorthand notation $Q(x_1^n)=Q(X_1^n=x_1^n)$. The
  expectation of random variable $X$ with respect to $Q$ is written
  $\sred_Q X$ and the variance is $\var_Q X$.
\item Moreover, measure $Q$ will be compared with codes, which for
  uniformity of notation will be represented in our approach as
  incomplete measures $P$, i.e., a code $P$ in our approach is a
  function that satisfies $P(x_1^n)\ge 0$ and the Kraft inequality
  $\sum_{x_1^n} P(x_1^n)\le 1$.\footnote{Here we deviate from the
    standard definition. Usually a code is a function
    $C:\mathbb{X}^*\rightarrow\klam{0,1}^*$ with length
    $\abs{C(x_1^n)}$, where $2^{-\abs{C(x_1^n)}}$ is a certain code in
    our sense.  An instance of such a code $C$ is the Lempel-Ziv code
    \cite{ZivLempel77}.}  Subsequently, for a code $P$ (or for measure
  $Q$), we define the pointwise mutual information between blocks
\begin{align}
  I^P(n)=-\log P(X_{-n+1}^0)-\log P(X_1^n)+\log P(X_{-n+1}^n)
  .
\end{align}
In the formula $\log$ stands for the binary logarithm.
\end{enumerate}

Now we may define the Hilberg exponents:
\begin{definition}[Hilberg exponents] Define the positive logarithm 
  \begin{align}
    \log^+ x
    =
    \begin{cases}
      \log (x+1), & x\ge 0, \\
      0, & x<0.
    \end{cases}
  \end{align}
For a~code $P$ we introduce
\begin{align}
  \gamma^+_P&=\limsup_{n\rightarrow\infty}
  \frac{\log^+ I^P(n)}{\log n} 
  ,
  \\
  \gamma^-_P&=\liminf_{n\rightarrow\infty}
  \frac{\log^+ I^P(n)}{\log n} 
  ,
  \\
  \delta^+_P&=\limsup_{n\rightarrow\infty}
  \frac{\log^+ \sred_Q I^P(n)}{\log n} 
  ,
  \\
  \delta^-_P&=\liminf_{n\rightarrow\infty}
  \frac{\log^+ \sred_Q I^P(n)}{\log n} 
  .
\end{align}
The above numbers will be called: $\gamma^+_P$---the upper random
Hilberg exponent, $\gamma^-_P$---the lower random Hilberg exponent,
$\delta^+_P$---the upper expected Hilberg exponent, and
$\delta^-_P$---the lower expected Hilberg exponent.
\end{definition}

Exponents $\gamma^\pm_P$ are random variables, whereas $\delta^\pm_P$
are constants.  By definition,
\begin{align}
\gamma^+_P\ge\gamma^-_P&\ge 0
,
\\
\delta^+_P\ge\delta^-_P&\ge 0
.
\end{align}
Let us remark that Hilberg exponents for $P=Q$ quantify some sort of
long-range non-Markovian dependence in the process. In particular, for
$Q$ being the measure of an IID process or a hidden Markov process
with a finite number of hidden states, mutual information $\sred_Q
I^Q(n)$ is zero or bounded, respectively, and hence
$\delta^\pm_Q=\gamma^\pm_Q=0$. The same is true for $k$-parameter
sources since $I^Q(n)$ is proportional to $k\log n$
\cite{Atteson99,BarronRissanenYu98,LiYu00}.  However, if information
$I^Q(n)$ grows proportionally to $n^\beta$ where $\beta\in[0,1]$ then
$\gamma^\pm_Q=\delta^\pm_Q=\beta$.  There exist some simple
non-Markovian but still mixing sources \cite{Debowski12}, being a
generalization of the Santa Fe processes, for which $\delta^\pm_Q$ can
be an arbitrary number in range $(0,1)$.

The are a few reasons why we introduce so many Hilberg exponents, for
each stationary measure $Q$ and each code $P$. The first one is that
the pointwise mutual information $I^P(n)$ may grow by leaps and
bounds. Consequently, the upper bounding power-law function may rise
faster than the respective lower bounding power-law function. This may
happen indeed. The second reason is that, a priori, different rates of
growth might be observed for the pointwise and the expected mutual
information. If they are equal, this should be separately proved. As
for the final reason, whereas we are here most interested in case
$P=Q$, some reason for investigating the pointwise mutual information
$I^P(n)$ for $P\neq Q$ is that, paradoxically, sometimes it is easier
to say something about $Q$-typical behavior of $I^P(n)$ than about
$I^Q(n)$. As we will show, this concerns not only statistical
applications, where we do not know $Q$, but also some theoretical
results, where $Q$ is known. Thus our definition is not too generic.

In the following we will show that Hilberg exponents satisfy a number
of relationships that impose some order among them. The first thing we
note are inequalities
\begin{align}
\gamma^-_P\le\gamma^+_P&\le 1
,
\\
\delta^-_P\le\delta^+_P&\le 1
,
\end{align}
which hold in the following cases:
\begin{enumerate}
\item For $P=Q$: Let $k(n)$ and $l(n)$ be nondecreasing functions of
  $n$, where $k(n)+l(n)\rightarrow\infty$. By an easy generalization of
  the Shannon-McMillan-Breiman theorem
  \cite{Breiman57,Chung61,AlgoetCover88}, we have $Q$-almost surely
  that
  \begin{align}
    \label{SMB}
    \lim_{n\rightarrow\infty} \frac{1}{k(n)+l(n)+1} \kwad{-\log Q(X_{-k(n)}^{l(n)})}
    =
    h_Q
    ,
  \end{align}
  where $h_Q$ is the entropy rate of measure $Q$ ($h_Q$ is a random
  variable if $Q$ is nonergodic). Hence $\lim_{n\rightarrow\infty}
  I^Q(n)/n=0$ and so $\gamma^-_Q\le\gamma^+_Q\le 1$ holds $Q$-almost
  surely.  Moreover, by stationarity,
  \begin{align}
    \lim_{n\rightarrow\infty} \frac{1}{k(n)+l(n)+1} \sred_Q
    \kwad{-\log Q(X_{-k(n)}^{l(n)})}
    =
    \sred_Q h_Q
    .
  \end{align}
  Hence $\lim_{n\rightarrow\infty} \sred_Q I^Q(n)/n=0$ and so 
  $\delta^-_Q\le\delta^+_Q\le 1$.
\item For $P$ being universal almost surely and in expectation: Let
  $k(n)$ and $l(n)$ be nondecreasing functions of $n$, where
  $k(n)+l(n)\rightarrow\infty$, as previously. Here, we will say that
  a code $P$ is universal almost surely if for every stationary
  distribution $Q$ we have $Q$-almost surely
  \begin{align}
    \label{StronglyUniversal}
    \lim_{n\rightarrow\infty} \frac{1}{k(n)+l(n)+1} \kwad{-\log P(X_{-k(n)}^{l(n)})}
    =
    h_Q
    ,
  \end{align}
  where $h_Q$ is the entropy rate of measure $Q$. In that case
  $\lim_{n\rightarrow\infty} I^P(n)/n=0$ so
  $\gamma^-_P\le\gamma^+_P\le 1$ holds $Q$-almost surely. Moreover,
  we will say that a code $P$ is universal in expectation if for every
  stationary distribution $Q$ we have
   \begin{align}
     \label{WeaklyUniversal}
    \lim_{n\rightarrow\infty} \frac{1}{k(n)+l(n)+1} \sred_Q
    \kwad{-\log P(X_{-k(n)}^{l(n)})}
    =
    \sred_Q h_Q
    .
  \end{align}
  In that case $\lim_{n\rightarrow\infty} \sred_Q I^P(n)/n=0$ so 
  $\delta^-_P\le\delta^+_P\le 1$.
\end{enumerate}
Equality (\ref{StronglyUniversal}) can be satisfied since stationary
ergodic measures are mutually singular. Moreover, almost surely
universal codes exist if and only if the alphabet $\mathbb{X}$ is
finite \cite{Kieffer78}. Some example of an almost surely universal
code is the Lempel-Ziv code \cite{ZivLempel77}. We also note that an
almost surely universal code $P$ is universal in expectation if $-\log
P(X_{-k(n)}^{l(n)})\le C(k(n)+l(n)+1)$ for some constant $C>0$
\cite{Weissman05}. In particular, the Lempel-Ziv code satisfies this
inequality.

As we have indicated, there are quite many Hilberg exponents, for
different measures and for different codes. Seeking for some order in
this menagerie, we may look for results of three kinds:
\begin{enumerate}
\item For a fixed code $P$ and a measure $Q$, we relate the random
  exponents $\gamma^\pm_P$ and the expected exponents $\delta^\pm_P$.
\item For two codes $P$ and $R$, we relate the exponents of a fixed
  kind, say $\delta^\pm_P$ and $\delta^\pm_R$ for some measure $Q$.
\item For a fixed code $P$ and a measure $Q$, we directly evaluate
  exponents $\gamma^\pm_P$ and $\delta^\pm_P$.
\end{enumerate}
In the following we will present some results of these three sorts.
They have varying weight but they shed some light onto unknown
territory.

The first kind of results could be called ``second-order'' analogues
of the Shannon-McMillan-Breiman (SMB) theorem (\ref{SMB}). The
original idea of the SMB theorem was to relate the asymptotic growth
of pointwise and expected entropies for an ergodic process $Q$ with
$P=Q$. With some partial success, this idea was then extended to the
case when the code $P$ was different to the underlying measure $Q$
\cite{Kieffer73,Barron85,Orey85}.  In contrast, relating the random
Hilberg exponents $\gamma^\pm_P$ and the expected Hilberg exponents
$\delta^\pm_P$ means relating the speed of growth of the pointwise and
expected mutual informations, which are differences of the respective
entropies. This is a somewhat subtler effect than the SMB theorem,
hence our ``second-order'' terminology.  In this domain we have
achieved an interesting result.  For an arbitrary code $P$ with
exponent $\delta^-_P>0$, let us introduce parameter
\begin{align}
  \label{Gap}
  \epsilon_P&=\limsup_{n\rightarrow\infty}
    \frac{\log^+ \kwad{\var_Q I^P(n)/\sred_Q I^P(n)}}{\log n}
    .
\end{align}
Our result is a sandwich bound for random Hilberg exponents in terms
of the expected Hilberg exponents for $P=Q$:
\begin{theorem}
  \label{theoMeasureErgodicAnalogue}
  For an ergodic measure $Q$ over a finite alphabet, random Hilberg
  exponents $\gamma^\pm_Q$ are almost surely constant. Moreover, we
  have $Q$-almost surely
  \begin{align}
    \label{SMBMIQ1}
    \delta^+_Q&\ge\gamma^+_Q\ge\delta^+_Q-\epsilon_Q
    ,
    \\
    \label{SMBMIQ2}
    \delta^-_Q&\ge\gamma^-_Q\ge\delta^-_Q-\epsilon_Q
    ,
  \end{align}
  where the left inequalities hold without restrictions, whereas the
  right inequalities hold for $\delta^-_Q>0$.
\end{theorem}
As we have written, this theorem may be considered an analogue of the
SMB theorem for the mutual information of the underlying measure.

We cannot refrain from mentioning the uncommon technique that has led
us to proving Theorem \ref{theoMeasureErgodicAnalogue}.  It is
remarkable that this result can be demonstrated without invoking the
ergodic theorem. Instead, we use an auxiliary ``Kolmogorov code''
\begin{align}
\label{KolmogorovCode}
S(x_1^n)=2^{-K(x_1^n|F)},
\end{align}
where $K(x_1^n|F)$ is the prefix-free Kolmogorov complexity of a
string $x_1^n$ given an object $F$ on an additional infinite tape
\cite{Chaitin75,LiVitanyi08}. The object $F$ can be another string or,
in our application, a definition of some measure. Respectively,
quantity $I^S(n)$ is the conditional algorithmic mutual
information. By some approximate translation invariance of Kolmogorov
complexity, we can show that the random Hilberg exponents
$\gamma^\pm_S$ are almost surely constant on ergodic sources $Q$.
This fact constitutes a novel contribution of algorithmic information
theory to the study of stochastic processes. Further, using Markov
inequality and Borel-Cantelli lemma, we can show that
$\gamma^\pm_S\le\delta^\pm_S$ on ergodic sources $Q$, as well. To
complete the rough idea of the proof of Theorem
\ref{theoMeasureErgodicAnalogue}, let us mention that
$\gamma^\pm_S=\gamma^\pm_Q$ and $\delta^\pm_S=\delta^\pm_Q$ if we
condition the Kolmogorov complexity on the distribution $Q$, i.e., if
we plug $F=Q$ in (\ref{KolmogorovCode}). This follows by Shannon-Fano
coding and Barron inequality \cite[Theorem 3.1]{Barron85b}. In this
way we obtain the left inequalities in
(\ref{SMBMIQ1})--(\ref{SMBMIQ2}). The right inequalities are
demonstrated in quite a similar fashion, using some auxiliary
quantities for the Kolmogorov code $S$, which we will call inverse
Hilberg exponents.

Now let us proceed to the second kind of results, those for two
codes. Here our results are modest. Let us note that in applications
we often do not know the underlying measure and we cannot compute the
Kolmogorov complexity but we can compute some other universal codes
such as the Lempel-Ziv code. Thus it would be advisable to relate
Hilberg exponents for computable (in the sense of the theory of
computation) universal codes to Hilberg exponents for the Kolmogorov
code $S$ or the underlying measure $Q$. For a universal code, we may
suppose that the longer the code is, the larger Hilberg exponents it
has. This hope is partly confirmed by the following simple theorem.
\begin{theorem}
\label{theoCompareCode}
Let $f_n$ be such that
\begin{align}
  \limsup_{n\rightarrow\infty}
  \frac{\log^+ \abs{f_n}}{\log n}=0
  .
\end{align}
Suppose that for codes $P$ and $R$ and a stationary measure $Q$ we
have
\begin{align}
  \label{PseudoBarron}
  \sred_Q \kwad{-\log P(X_1^n)}&\le \sred_Q \kwad{-\log R(X_1^n)}+f_n
  ,
  \\
  \label{AlmostUniversal}
  \lim_{n\rightarrow\infty}  \frac{1}{n}\sred_Q \kwad{-\log P(X_1^n)}
  &=
  \lim_{n\rightarrow\infty}  \frac{1}{n}\sred_Q \kwad{-\log R(X_1^n)}
  .
\end{align}
Then
\begin{align}
  \label{CompareCode}
  \delta^+_R\ge \delta^-_P
  .
\end{align}
\end{theorem}

The simple proof of the above proposition rests on this lemma:
\begin{lemma}[\mbox{\cite{Debowski11b}}]
\label{theoExcess}
Consider a~function $G:\mathbb{N}\rightarrow\mathbb{R}$ such that
$\lim_k G(k)/k=0$ and $G(n)\ge 0$ for all but finitely many $n$. For
infinitely many $n$, we have $2G(n)- G(2n)\ge 0$.
\end{lemma}
To prove Theorem \ref{theoCompareCode}, it suffices to put $G(n)=-\log
R(X_1^n)+\log P(X_1^n)+f_n$ and use subadditivity of the function
$\log^+$, i.e., inequality
\begin{align}
  \label{LogSubadditivity}
  \log^+(x+y)\le\log^+x+\log^+y
  .
\end{align}

The most useful applications of Theorem \ref{theoCompareCode} are as
follows:  Condition (\ref{PseudoBarron}) is satisfied, with $f_n=2\log
n+C$, where $C>0$, for any computable code $R$ and unconditional
Kolmogorov code $P(x_1^n)=2^{-K(x_1^n)}$, where $K(x_1^n)$ is the
unconditional prefix-free Kolmogorov complexity of a string
$x_1^n$. Moreover condition (\ref{PseudoBarron}) is satisfied, with
$f_n=0$, for any code $R$ and $P=Q$.  Condition
(\ref{AlmostUniversal}) is satisfied if $P$ and $R$ are universal in
expectation or if $R$ is universal in expectation and $P=Q$. Moreover,
conditions (\ref{PseudoBarron}) and (\ref{AlmostUniversal}) are
satisfied with $f_n=0$ for $R=Q$ and $P=E$ \cite{GrayDavisson74b},
where $E$ is the random ergodic measure given by
\begin{align}
  \label{RandomErgodic}
  E=Q(\cdot|\mathcal{I})
\end{align}
where $\mathcal{I}$ is the shift-invariant algebra
\cite{Kallenberg97,Debowski11b}.

Let us remark that inequality (\ref{CompareCode}) is not very
strong. We have been able to relate only the upper expected Hilberg
exponents with the lower expected Hilberg exponent. It would be more
interesting if, for two different codes, we were able to compare the
random exponents of the same kind, i.e., an upper exponent with an
upper exponent and a lower exponent with a lower exponent. This
requires relating functions $I^P(n)$ and $I^R(n)$. But even relatively
simple cases, such as comparing $I^Q(n)$ and $I^E(n)$, where $E$ is
the random ergodic measure (\ref{RandomErgodic}), are not trivial. In
that case, $I^Q(n)-I^E(n)$ equals triple mutual information between
two blocks and the shift-invariant algebra, which may be
negative. Therefore we leave strenghtening Theorem
\ref{theoCompareCode} as an open problem.

To end this introduction, we mention the third kind of results, which
concern analytic evaluation of Hilberg exponents for concrete examples
of processes. As we have indicated, the expected Hilberg exponents for
IID processes, Markov processes, and hidden Markov processes do vanish
because the mutual information is either zero or bounded. To
investigate less trivial cases, in this paper, we will consider three
kinds of other sources. These will be: the mixture Bernoulli process,
the original Santa Fe processes mentioned in the very beginning of
this paper \cite{Debowski11b,Debowski12}, and some modification of the
Santa Fe processes. It should be noted that all analyzed processes are
conditionally IID.  The pointwise mutual information $I^Q(n)$ for such
sources is equal to a difference of redundancies. As shown in
\cite{Atteson99,BarronRissanenYu98,LiYu00}, in case of $k$-parameter
processes, such as the mixture Bernoulli process, the redundancy grows
proportionally to $k\log n$. For this reason all Hilberg exponents do
vanish for the mixture Bernoulli process. For completeness we will
reproduce the relevant simple calculation in this paper. In contrast,
the second example, the original Santa Fe process, exhibits a stronger
dependence, namely $I^Q(n)$ grows proportionally to $n^\beta$, where
$\beta\in(0,1)$ is a certain free parameter of the process.  Therefore
all four Hilberg exponents are equal to $\beta$. In the third example,
we can, however, modify the definition of the Santa Fe process so that
the upper expected exponent is $\delta^+_Q=\beta$ for an arbitrary
parameter $\beta\in(0,1)$ whereas the lower expected exponent is
$\delta^-_Q=0$. This shows that upper and lower Hilberg exponents need
not be equal.

The further contents of the paper is as follows. In Section
\ref{secKolmogorov}, we discuss properties of Hilberg exponents for
the Kolmogorov code $S$.  In Section \ref{secMeasure}, we translate
these results for the underlying measure $Q$, proving Theorem
\ref{theoMeasureErgodicAnalogue}.
% In the very short Section \ref{secOther}, we prove Theorem
% \ref{theoCompareCode}.
Finally, in Section \ref{secEvaluation}, we evaluate Hilberg exponents
for the mixture Bernoulli process and the Santa Fe processes.

\section{Kolmogorov code}
\label{secKolmogorov}

% A prominent role in our theory will be played by the Kolmogorov code
% \begin{align}
% \label{KolmogorovCode}
% S(x_1^n)=2^{-K(x_1^n|F)},
% \end{align}
% where $K(x_1^n|F)$ is the prefix-free Kolmogorov complexity of a
% string $x_1^n$ given an object $F$ on an additional infinite tape
% \cite{Chaitin75,LiVitanyi08}. The object $F$ can be another string, or
% a definition of a some measure $Q$, what will appear very useful in
% the next section.  

A prominent role in our demonstration of Theorem
\ref{theoMeasureErgodicAnalogue} will be played by the Kolmogorov code
(\ref{KolmogorovCode}), whose properties will be discussed in this
section.  Although the results of this section are used in the next
section to prove Theorem \ref{theoMeasureErgodicAnalogue}, they can be
regarded as facts of some independent interest. For this reason, we
assemble them into a separate narrative unit.

To begin with, let us note that, for a finite alphabet $\mathbb{X}$, the
Kolmogorov code is universal almost surely and in expectation, simply
because it is dominated by the Lempel-Ziv code.  Independently,
universality of the Kolmogorov code has been previously shown by
Brudno \cite{Brudno83} in the context of dynamical systems.  Although
Kolmogorov complexity itself is incomputable, the Hilberg exponents
for the Kolmogorov code can be evaluated in some cases and enjoy a few
nice properties. These properties stem from the fact that function
$I^S(n)$ equals the algorithmic mutual information
\begin{align}
  \label{AMI}
  I^S(n)= K(X_{-n+1}^{0}|F)+K(X_{1}^{n}|F)-K(X_{-n+1}^{n}|F) 
  ,
\end{align}
an important concept in algorithmic information theory.

Now let us present some new results. 
% The first simple fact is that
% Hilberg exponents for the Kolmogorov code can be upper bounded for
% sources over a finite alphabet.
% \begin{theorem}
%   \label{theoKolmogorovFinite}
%   Consider code (\ref{KolmogorovCode}) and a stationary measure $Q$
%   over a finite alphabet $\mathbb{X}$. We have
%   $\gamma^\pm_S,\delta^\pm_S\le 1$.
% \end{theorem}
% \begin{proof}
%   For a finite alphabet we have $K(x_1^n|F)\le Cn$, where $C>0$.  
%   Hence $I^S(n)\le 2Cn$, which implies the claim.
% \end{proof}
A simple but important fact is that the random Hilberg exponents are
almost surely constant on ergodic sources.  This fact is a consequence
of approximate shift-invariance of Kolmogorov complexity. That
property seemingly has not been noticed so far and it provides an
interesting link between algorithmic complexity and ergodic theory.
\begin{theorem}
  \label{theoKolmogorovErgodic}
  Consider code (\ref{KolmogorovCode}) and an ergodic measure $Q$ over
  a finite alphabet $\mathbb{X}$. Exponents $\gamma^-_S$ and
  $\gamma^+_S$ are $Q$-almost surely constant.
\end{theorem}
\emph{Remark:} The random Hilberg exponents can be different for
different ergodic sources, so for nonergodic sources they can be
random.
\begin{proof}
  For $t>0$, from the shortest program that computes $x_1^n$, we can
  construct a program that computes $x_{t+1}^{t+n}$, whose length
  exceeds the length of the program for $x_1^n$ no more than
  $K(x_{n+1}^{n+t}|F)+C$, where $C>0$. Analogously, from the shortest
  program that computes $x_{t+1}^{t+n}$, we can construct a program
  that computes $x_1^n$, whose length exceeds the length of the
  program for $x_{t+1}^{t+n}$ no more than $K(x_{1}^{t}|F)+C$.  This
  yields
  \begin{align*}
    \abs{K(x_1^n|F)-K(x_{t+1}^{t+n}|F)}\le
    \max\klam{K(x_{1}^{t}|F),K(x_{n+1}^{n+t}|F)}+C
    .
  \end{align*}
  Thus
  \begin{align*}
    &\abs{I^S(n)-
      \kwad{K(X_{t-n+1}^{t}|F)-K(X_{t+1}^{t+n}|F)+K(X_{t-n+1}^{t+n}|F)}}
    \\
    &\qquad
    \le
    3\max\klam{K(X_{-n+1}^{-n+t}|F),K(X_{1}^{t}|F),K(X_{n+1}^{n+t}|F)}+3C
    .
  \end{align*}
  Now we notice that for a finite alphabet we have
  \begin{align*}
    K(X_{-n+1}^{-n+t}|F),K(X_{1}^{t}|F),K(X_{n+1}^{n+t}|F)\le Ct
    ,
  \end{align*}
  where $C>0$.  Hence by inequality~(\ref{LogSubadditivity}) functions
  $\gamma^-_S$ and $\gamma^+_S$ are shift-invariant. Since $Q$ is
  ergodic, it means they must be $Q$-almost surely constant.
\end{proof}

Subsequently, we will give some bounds for the random Hilberg
exponents in terms of the expected Hilberg exponents. To achieve this
goal, we need two lemmas and some additional definition. In the
following, we will write $a(n)\pge b(n)$ if $a(n)+C\ge b(n)$ for all
arguments $n$ and a $C\ge 0$, whereas $a(n)\ple b(n)$ if $a(n)\le
b(n)+C$ under the same conditions. We also write $a(n)\peq b(n)$ if
$a(n)\pge b(n)$ and $a(n)\ple b(n)$.

The following lemma is the first step on our way. It says that mutual
information $I^S(n)$ is almost a nondecreasing function.
\begin{lemma}
  \label{theoKolmogorovGrowing}
  Consider code (\ref{KolmogorovCode}). For all $m\ge 1$, we have
\begin{align}
  \label{ISGrowing}
  I^S(n+m)
  &\pge
  I^S(n)
  - 4\log m
  .
\end{align}
\end{lemma}
\begin{proof}
  For strings $u$ and $v$, denote the algorithmic mutual information
  \begin{align}
    \label{AMI2}
    I(u:v|F)=K(u|F)+K(v|F)-K(u,v|F)
    .
  \end{align}
  We have $I(u:v|F)\pge 0$ \cite{LiVitanyi08}.
  Concatenating strings decreases their complexity. Namely,
  $$K(uv|F)\ple K(u,v|F)\ple
  K(uv|F)+K(\abs{v}|uv,F),$$ where $\abs{v}$ is the
  length of $v$. Hence
  \begin{align*}
    I^S(n+m)+4\log m\pge I(a,b:c,d|F)
  \end{align*}
  whereas
  \begin{align*}
    I^S(n)\peq I(b:c|F)
  \end{align*}
  where $a=X_{-n+m+1}^{-n}$, $b=X_{-n+1}^{0}$, $c=X_{1}^{n}$, and $d=X_{n+1}^{m}$.

  Using identity $K(u,v|F)\peq K(u|F)+K(v|u,K(u|F),F)$
  \cite{LiVitanyi08}, we can further show the data processing
  inequality for the algorithmic mutual information,
  \begin{align*}
    I(a,b:c,d|F)
    &\peq I(a,b:c|F)+I(a,b:d|c,K(c|F),F)
    \\
    &\pge I(a,b:c|F)
    \\
    &\peq I(b:c|F)+I(a:c|b,K(b|F),F)
    \\
    &\pge I(b:c|F)
    ,
  \end{align*}
  which proves the claim.
\end{proof}

With the above lemma we can prove another auxiliary result. This
result says that Hilberg exponents for the Kolmogorov code can be
defined using only a subsequence of exponentially growing block
lengths.
\begin{lemma}
\label{theoKolmogorovSubsequence}
Consider code (\ref{KolmogorovCode}). We have
\begin{align}
  \label{UpperRandomS}
  \gamma^+_S&=\limsup_{k\rightarrow\infty}
  \frac{\log^+ I^S(2^k)}{\log 2^k} 
  ,
  \\
  \label{LowerRandomS}
  \gamma^-_S&=\liminf_{k\rightarrow\infty}
  \frac{\log^+ I^S(2^k)}{\log 2^k} 
  ,
  \\
  \label{UpperExpectedS}
  \delta^+_S&=\limsup_{k\rightarrow\infty}
  \frac{\log^+ \sred_Q I^S(2^k)}{\log 2^k} 
  ,
  \\
  \label{LowerExpectedS}
  \delta^-_S&=\liminf_{k\rightarrow\infty}
  \frac{\log^+ \sred_Q I^S(2^k)}{\log 2^k} 
  .
\end{align}
\end{lemma}
\begin{proof}
  By Lemma \ref{theoKolmogorovGrowing}, for $n=2^k+m$, where $0\le m<
  2^k$, we have
    \begin{align*}
      I^S(n)\le I^S(2^{k+1})+4\log (2^k-m)+C\le I^S(2^{k+1})+4k+C
      .
    \end{align*}
    Thus
    \begin{align*}
      \limsup_{n\rightarrow\infty}
      \frac{\log^+ I^S(n)}{\log n}
      &\le
      \limsup_{k\rightarrow\infty}
      \frac{\log^+ (I^S(2^{k+1})+4k+C)}{\log 2^k}
      \\
      &\le
      \limsup_{k\rightarrow\infty}
      \frac{\log^+ I^S(2^{k+1})}{\log 2^k}
      +
      \limsup_{k\rightarrow\infty}
      \frac{\log^+ (4k+C)}{\log 2^k}
      \\
      &=
      \limsup_{k\rightarrow\infty}
      \frac{\log^+ I^S(2^{k})}{\log 2^k}
      .
    \end{align*}
    This proves (\ref{UpperRandomS}) since trivially we have a
    converse inequality. 

    Now let $n=2^k+m$, where $0< m\le 2^k$. From (\ref{ISGrowing}), we
    obtain
    \begin{align*}
      I^S(n)\ge I^S(2^{k})-4\log m-C\ge I^S(2^{k})-4k-C
      .
    \end{align*}
    Thus
    \begin{align*}
      \liminf_{n\rightarrow\infty}
      \frac{\log^+ I^S(n)}{\log n}
      &\ge
      \liminf_{k\rightarrow\infty}
      \frac{\log^+ (I^S(2^{k})-4k-C)}{\log 2^{k+1}}
      \\
      &\ge
      \liminf_{k\rightarrow\infty}
      \frac{\log^+ I^S(2^{k})}{\log 2^{k+1}}
      -
      \limsup_{k\rightarrow\infty}
      \frac{\log^+ (4k+C)}{\log 2^{k+1}}
      \\
      &=
      \liminf_{k\rightarrow\infty}
      \frac{\log^+ I^S(2^{k})}{\log 2^k}
      .
    \end{align*}
    This proves (\ref{LowerRandomS}) for trivially we have a converse
    inequality.

    The proofs of (\ref{UpperExpectedS}) and (\ref{LowerExpectedS})
    are analogous.
\end{proof}

To finish the preparations, we need some auxiliary concept. Recall
that algorithmic mutual information (\ref{AMI2}) is greater than a
constant. Using this result, for the Kolmogorov code we can introduce
the following inverse Hilberg exponents.
\begin{definition}[inverse Hilberg exponents]
  Consider code (\ref{KolmogorovCode}).  Let $B$ be such that
  $I^S(n)+B\ge 1$. Define
  \begin{align}
    \zeta^+_S&=\limsup_{n\rightarrow\infty}
    \frac{\log^+ \kwad{\sred_Q (I^S(n)+B)^{-1}}^{-1}}{\log n}
    ,
    \\
    \zeta^-_S&=\liminf_{n\rightarrow\infty}
    \frac{\log^+ \kwad{\sred_Q (I^S(n)+B)^{-1}}^{-1}}{\log n}     
    .
  \end{align}
  The above numbers will be called: $\zeta^+_S$---the upper inverse
  expected Hilberg exponent and $\zeta^-_S$---the lower inverse
  expected Hilberg exponent.
\end{definition}
We have $\zeta^+_S\ge \zeta^-_S\ge 0$, whereas $\delta^+_S\ge
\zeta^+_S$ and $\delta^-_S\ge \zeta^-_S$ by the Jensen inequality
$\sred_Q X\ge \kwad{\sred_Q X^{-1}}^{-1}$ for $X> 0$.

Now we may state and prove the theorem which links the expected and
the random Hilberg exponents for the Kolmogorov code. It will be first
stated and proved for general stationary (not necessarily ergodic)
measures over a countable alphabet. Subsequently, we will present a
corollary for ergodic measures and some strengthening for a finite
alphabet.
\begin{theorem}
\label{theoKolmogorov}
Consider code (\ref{KolmogorovCode}) and an arbitrary stationary
measure $Q$. Then:
\begin{enumerate}
\item $\delta^+_S\ge\gamma^+_S$ $Q$-almost surely and $\esssup_Q
  \gamma^+_S\ge\zeta^+_S$.
\item $\delta^-_S\ge\essinf_Q \gamma^-_S$ and $\gamma^-_S\ge\zeta^-_S$
  $Q$-almost surely.
\end{enumerate}
\end{theorem}
\begin{proof}
  \begin{enumerate}
  \item Let $\epsilon>0$. Observe that from the Markov inequality we
    have
    \begin{align*}
      \sum_{k=1}^\infty
      Q\okra{\frac{I^S(2^k)+B}{(2^k)^{\delta^+_S+\epsilon}}\ge 1}
      &\le
      \sum_{k=1}^\infty
      \frac{\sred_Q I^S(2^k)+B}{(2^k)^{\delta^+_S+\epsilon}}
      \\
      &\le
      A 
      +
      \sum_{k=1}^\infty
      \frac{(2^k)^{\delta^+_S+\epsilon/2}}{(2^k)^{\delta^+_S+\epsilon}}
      <\infty
      ,
    \end{align*}
    where $A<\infty$. Hence, by the Borel-Cantelli lemma, we have
    $Q$-almost surely
    \begin{align*}
      \limsup_{k\rightarrow\infty}
      \frac{\log^+ (I^S(2^k)+B)}{\log 2^k} 
      \le \delta^+_S+\epsilon
      .
    \end{align*}
    By arbitrariness of $\epsilon$ and by
    inequality~(\ref{LogSubadditivity}), the bound is true with $\epsilon=0$
    and $B=0$, which implies $\delta^+_S\ge\gamma^+_S$ $Q$-almost
    surely by Lemma~\ref{theoKolmogorovSubsequence}.

    Now we will prove that $\esssup_Q \gamma^+_S\ge\zeta^+_S$. Denote
    $\esssup_Q \gamma^+_S=\beta$ and let $\epsilon>0$. Then
    $Q(\gamma^+_S>\beta+\epsilon/2)=0$ whence
    \begin{align}
      \label{IOGreater}
      Q\okra{I^S(n)+B\ge n^{\beta+\epsilon} \text{ infinitely often}}=0
      .
    \end{align}
    Denote $p(n)=Q\okra{I^S(n)+B<n^{\beta+\epsilon}}$. We have
    \begin{align*}
      \sred_Q \okra{I^S(n)+B}^{-1}\ge n^{-\beta-\epsilon} p(n).
    \end{align*}
    By (\ref{IOGreater}), $\lim_{n\rightarrow\infty}p(n)=1$. Hence
    $\zeta^+_S\le \beta+\epsilon$. Since $\epsilon$ was arbitrary,
    this implies the claim.

  \item The proof is analogous to the proof of (i). Write
    $\essinf_Q \gamma^-_S=\beta$ and let $\epsilon>0$. Then
    $Q(\gamma^-_S<\beta-\epsilon/2)=0$ whence
    \begin{align}
      \label{IOLess}
      Q\okra{I^S(n)\le n^{\beta-\epsilon} \text{ infinitely often}}=0
      .
    \end{align}
    Denote $p(n)=Q\okra{I^S(n)>n^{\beta-\epsilon}}$. We have
    \begin{align*}
      \sred_Q I^S(n)\ge n^{\beta-\epsilon} p(n).
    \end{align*}
    By (\ref{IOLess}), $\lim_{n\rightarrow\infty}p(n)=1$. Thus
    $\delta^-_S\ge \beta-\epsilon$. Since $\epsilon$ was arbitrary,
    this implies $\delta^-_S\ge\essinf_Q \gamma^-_S$.

    Now we will show the second claim. Let $\epsilon>0$. From the
    Markov inequality
    \begin{align*}
      \sum_{k=1}^\infty
      Q\okra{\frac{I^S(2^k)+B}{(2^k)^{\zeta^-_S-\epsilon}}\le 1}
      &\le
      \sum_{k=1}^\infty
      \frac{\sred_Q (I^S(2^k)+B)^{-1}}{(2^k)^{-\zeta^-_S+\epsilon}}
      \\
      &\le
      A 
      +
      \sum_{k=1}^\infty
      \frac{(2^k)^{-\zeta^-_S+\epsilon/2}}{(2^k)^{-\zeta^-_S+\epsilon}}
      <\infty
      ,
    \end{align*}
    where $A<\infty$. Thus, by the Borel-Cantelli lemma, $Q$-almost
    surely
    \begin{align*}
      \liminf_{k\rightarrow\infty}
      \frac{\log^+ (I^S(2^k)+B)}{\log 2^k} 
      \ge \zeta^-_S-\epsilon
      .
    \end{align*}
    As in (i), we may put $\epsilon=0$ and $B=0$, whence
    $\gamma^-_S\ge\zeta^-_S$ $Q$-almost surely follows by
    Lemma~\ref{theoKolmogorovSubsequence}.
  \end{enumerate}
\end{proof}

Let us also present some add-ons to Theorem \ref{theoKolmogorov}.
Using Theorem \ref{theoKolmogorovErgodic}, Theorem
\ref{theoKolmogorov} can be specialized for ergodic measures over a
finite alphabet in an interesting way, which will be used later.
\begin{corollary}
  \label{theoKolmogorovErgodicAnalogue}
  By Theorem \ref{theoKolmogorovErgodic}, for an ergodic measure $Q$
  over a finite alphabet, equalities $\gamma^+_S=\esssup_Q \gamma^+_S$
  and $\gamma^-_S=\essinf_Q \gamma^-_S$ hold $Q$-almost surely. Hence,
  $Q$-almost surely we have
\begin{align}
  \label{SMBMI1}
  \delta^+_S&\ge\gamma^+_S\ge\zeta^+_S
  ,
  \\
  \label{SMBMI2}
  \delta^-_S&\ge\gamma^-_S\ge\zeta^-_S
  .
\end{align}
\end{corollary}
It is remarkable that inequalities (\ref{SMBMI1}) and (\ref{SMBMI2})
are demonstrated without invoking the ergodic theorem. 

Let us observe one more simple fact.  Namely, for a finite alphabet
$\mathbb{X}$, the bound for the random Hilberg exponents given by
Theorem \ref{theoKolmogorov} can be slightly strengthened since
$\esssup_Q \gamma^+_S\ge \sred_Q \gamma^+_S$ and $\sred_Q
\gamma^-_S\ge \essinf_Q \gamma^-_S$.
\begin{theorem}
\label{theoKolmogorovFinite}
Consider code (\ref{KolmogorovCode}) and an arbitrary stationary
measure $Q$. Then:
\begin{enumerate}
\item $\sred_Q \gamma^+_S \ge \zeta^+_S$ if the alphabet $\mathbb{X}$
  is finite.
\item $\delta^-_S\ge \sred_Q \gamma^-_S$.
\end{enumerate}
\end{theorem}
\begin{proof}
  \begin{enumerate}
  \item Function $-\log$ is convex. Hence we can use the Fatou lemma
    and the Jensen inequality,
    \begin{align*}
      \sred_Q \gamma^+_S
      &=
      \sred_Q \limsup_{n\rightarrow\infty} \frac{\log (I^S(n)+B)}{\log n}
      \\
      &\ge
      \limsup_{n\rightarrow\infty} \sred_Q \frac{\log (I^S(n)+B)}{\log n}
      \\
      &\ge
      \limsup_{n\rightarrow\infty} \sred_Q \kwad{-\frac{\log (I^S(n)+B)^{-1}}{\log n}}
      \\
      &\ge
      \limsup_{n\rightarrow\infty} \kwad{-\frac{\log \sred_Q (I^S(n)+B)^{-1}}{\log n}}
      =
      \zeta^+_S
      ,
    \end{align*}
    since the functions under the limits are bounded above.
  \item Reasoning as above,
    \begin{align*}
      \sred_Q \gamma^-_S
      &=
      \sred_Q \liminf_{n\rightarrow\infty} \frac{\log (I^S(n)+B)}{\log n}
      \\
      &\le
      \liminf_{n\rightarrow\infty} \sred_Q \frac{\log (I^S(n)+B)}{\log n}
      \\
      &\le
      \liminf_{n\rightarrow\infty} \frac{\log \sred_Q (I^S(n)+B)}{\log n}
      =
      \delta^-_S
      ,
    \end{align*}
    since the functions under the limits are nonnegative.
  \end{enumerate}
\end{proof}

\section{The underlying measure}
\label{secMeasure}

In this section we will prove Theorem
\ref{theoMeasureErgodicAnalogue}, which provides a bound for the
random Hilberg exponents of the underlying measure $Q$ in terms of the
measure's expected Hilberg exponents. For this goal, we will use the
results of the previous section. Our technique rests on a few
observations. The first observation is that four out of six Hilberg
exponents for the Kolmogorov code are equal to the Hilberg exponents
for the underlying measure $Q$ if we use a special conditional
Kolmogorov code. In this code, the definition of measure $Q$ is fed to
the Turing machine on an additional infinite tape, i.e., $F=Q$. By the
Shannon-Fano coding and the Barron inequality, such a Kolmogorov code
is equal to the measure $Q$ in a sufficiently good approximation.
\begin{theorem}
\label{theoMeasure}
Consider code (\ref{KolmogorovCode})
with $F=Q$, where $Q$ is an arbitrary stationary measure. Then:
\begin{enumerate}
\item $\delta^-_S=\delta^-_Q$ and $\delta^+_S=\delta^+_Q$.
\item $\gamma^-_S=\gamma^-_Q$ and $\gamma^+_S=\gamma^+_Q$
  $Q$-almost surely.
\end{enumerate}
\end{theorem}
\begin{proof}
\begin{enumerate}
\item 
  The Shannon-Fano coding gives
  \begin{align}
    \label{ShannonFano}
    -\log S(x_1^n)=K(x_1^n|Q)\le -\log Q(x_1^n)+2\log n +C
  \end{align}
  for a constant $C>0$ \cite{GacsTrompVitanyi01}.  Hence from the
  source coding inequality $$\sred_Q H^S(n)\ge \sred_Q H^Q(n),$$ we
  obtain 
  \begin{align}
    \label{ExpectedBound}
    \abs{\sred_Q I^S(n)-\sred_Q I^Q(n)}\le 4\log n +2C.
  \end{align}
  Thus by inequality~(\ref{LogSubadditivity}), $\delta^-_S=\delta^-_Q$ and
  $\delta^+_S=\delta^+_Q$.
\item Observe that $Q$-almost surely we have the following Barron
  inequalities, viz.\ \cite[Theorem 3.1]{Barron85b},
  \begin{align*}
    \lim_{n\rightarrow\infty} 
    \kwad{-\log S(X_{-n+1}^{n})+\log Q(X_{-n+1}^{n})}&=\infty
    ,
    \\
    \lim_{n\rightarrow\infty} 
    \kwad{-\log S(X_1^n)+\log Q(X_1^n)}&=\infty
    ,
    \\
    \lim_{n\rightarrow\infty} 
    \kwad{-\log S(X_{-n+1}^{2n})+\log Q(X_{-n+1}^{n})}&=\infty    
    .
  \end{align*}
  Combining these facts with the Shannon-Fano coding
  (\ref{ShannonFano}) yields
  $$\abs{I^S(n)-I^Q(n)}\le
  4\log n +2C.$$ for sufficiently large $n$, $Q$-almost surely.  Thus
  by inequality~(\ref{LogSubadditivity}), $\gamma^-_S=\gamma^-_Q$ and
  $\gamma^+_S=\gamma^+_Q$ holds on a set of full measure.
\end{enumerate}
\end{proof}

Theorem \ref{theoMeasure} implies two more specific corollaries of an
independent interest.  The first result states that Hilberg exponents
for a computable measure $Q$ are equal to Hilberg exponents for
unconditional prefix-free Kolmogorov complexity.
\begin{corollary}
  \label{theoMeasureComputable}
  If measure $Q$ is computable then for code $P(x_1^n)=2^{-K(x_1^n)}$,
  where $K(x_1^n)$ is unconditional prefix-free Kolmogorov complexity
  of $x_1^n$, we have
$$I^P(n)\peq I^S(n),$$ where we use code
(\ref{KolmogorovCode}) with $F=Q$ again. This implies
$\gamma^-_Q=\gamma^-_S=\gamma^-_P$ and
$\gamma^-_Q=\gamma^+_S=\gamma^+_P$ $Q$-almost surely, whereas
$\delta^-_Q=\delta^-_S=\delta^-_P$ and
$\delta^+_Q=\delta^+_S=\delta^+_P$.
\end{corollary}
The second result concerns a nonergodic measure with a given ergodic
decomposition. It says that Hilberg exponents for this nonergodic
measure are almost surely constant on almost all ergodic components of
the measure.
\begin{corollary}
  \label{theoMeasureNonergodic}
  Suppose that measure $Q$ has the random ergodic measure $E$ given by
  (\ref{RandomErgodic}). We have $Q=\sred_Q E$, so by the properties
  of integral, any set of full $Q$-measure has full $E$-measure
  $Q$-almost surely.  This implies that for code
  (\ref{KolmogorovCode}) with $F=Q$, we have $\gamma^-_S=\gamma^-_Q$
  and $\gamma^+_S=\gamma^+_Q$ $E$-almost surely for $Q$-almost all
  values of measure $E$. By Theorem~\ref{theoKolmogorovErgodic}, in
  case of a finite alphabet, this means that $\gamma^-_Q$ and
  $\gamma^+_Q$ are $E$-almost surely constant for those values of
  measure $E$.
\end{corollary}

What lacks for the proof of Theorem \ref{theoMeasureErgodicAnalogue}
is a computable lower bound for the inverse Hilberg exponents, defined
in the previous section for the Kolmogorov code $S$.  For an arbitrary
code $P$ with $\delta^-_P>0$, let us introduce parameter $\epsilon_P$
given by formula (\ref{Gap}).  First, we will show that the difference
between the expected and the inverse Hilberg exponents
$\delta^\pm_S-\zeta^\pm_S$ is bounded by parameter $\epsilon_S$ and
then we will show that  $\epsilon_S=\epsilon_Q$ for $F=Q$.
\begin{theorem}
  \label{theoEpsilonOne}
  Consider code (\ref{KolmogorovCode}) and an arbitrary stationary
  measure $Q$.  If $\delta^-_S>0$ then $\zeta^+_S\ge
  \delta^+_S-\epsilon_S$ and $\zeta^-_S\ge \delta^-_S-\epsilon_S$.
 \end{theorem}
 \begin{proof}
   Let $\alpha\in (0,1)$.  By $I^S(n)+B\ge 1$ and by Markov inequality
   we obtain
   \begin{align*}
     \sred_Q (I^S(n)+B)^{-1} 
     &\le 
     Q\okra{\abs{I^S(n)-\sred_Q{I^S(n)}}\ge\alpha(\sred_Q{I^S(n)}+B)}
     \\
     &\phantom{=}
     +
     \frac{1}{(1-\alpha)(\sred_Q{I^S(n)}+B)}
     \\
     &\le
     \frac{\var_Q I^S(n)}{\alpha^2(\sred_Q I^S(n)+B)^2}+
     \frac{1}{(1-\alpha)(\sred_Q{I^S(n)}+B)}
     .
   \end{align*}
   Hence
   \begin{align*}
     \kwad{\sred_Q (I^S(n)+B)^{-1}}^{-1}
     \ge 
     (\sred_Q{I^S(n)}+B)
     \okra{\frac{\var_Q I^S(n)}{\alpha^2(\sred_Q
         I^S(n)+B)}+\frac{1}{(1-\alpha)}}^{-1}
     ,
   \end{align*}
   which implies the claim by $\log(x/y)=\log x-\log y$,
   $\delta^-_S>0$, and inequality~(\ref{LogSubadditivity}).
 \end{proof}

 Subsequently, we will prove that parameter $\epsilon_S$ for the
 conditional Kolmogorov code with $F=Q$ is equal to parameter
 $\epsilon_Q$ for the underlying measure.
\begin{theorem}
\label{theoEpsilonTwo}
Consider code (\ref{KolmogorovCode})
with $F=Q$, where $Q$ is an arbitrary stationary measure. 
If $\delta^-_Q>0$ then $\epsilon_S=\epsilon_Q$.
\end{theorem}
\begin{proof}
  Since $\delta^-_Q>0$, by inequality (\ref{ExpectedBound}), we obtain
\begin{align*}
  \epsilon_S=\limsup_{n\rightarrow\infty}
  \frac{\log^+ \kwad{\var_Q I^S(n)/\sred_Q I^Q(n)}}{\log n}
  .
\end{align*}
In the following, we have 
\begin{align*}
  \var_Q I^S(n)\in
  &\left[\okra{\sqrt{\var_Q I^Q(n)}-\sqrt{\var_Q (I^S(n)-I^Q(n))}}^2
  \right.,
  \\
  &\left.\okra{\sqrt{\var_Q I^Q(n)}+\sqrt{\var_Q (I^S(n)-I^Q(n))}}^2
  \right]
  .
\end{align*}
Thus, to show $\epsilon_S=\epsilon_Q$ it suffices to prove that
\begin{align*}
  \limsup_{n\rightarrow\infty} 
  \frac{\log^+ \var_Q (I^S(n)-I^Q(n))}{\log n}
  =0
  .
\end{align*}
To demonstrate the latter fact, we will use Shannon-Fano coding
(\ref{ShannonFano}) and a stronger version of Barron's inequality,
viz.\ \cite[Theorem 3.1]{Barron85b}, namely,
\begin{align}
  Q(-\log P(X_1^n)+\log Q(X_1^n)\le -m)\le 2^{-m}
  ,
\end{align}
which holds for an arbitrary code $P$. Hence we obtain
\begin{align*}
  Q(\abs{I^S(n)-I^Q(n)}\ge 4\log n +C +m)\le 2^{-m}
\end{align*}
for a certain constant $C$. Subsequently, this yields
\begin{align*}
  \sred_Q\okra{I^S(n)-I^Q(n)}^2
  &\le
  (4\log n +C)^2+
  \sum_{m=0}^\infty (4\log n +C+m+1)^22^{-m}
  \\
  &\le A(\log n)^2 +B 
\end{align*}
for certain $A,B>0$, which proves the claim.
\end{proof}

Now we may prove Theorem \ref{theoMeasureErgodicAnalogue}.
\begin{proof*}{Theorem \ref{theoMeasureErgodicAnalogue}}
  Apply Corollary \ref{theoKolmogorovErgodicAnalogue} from the
  previous section and Theorems \ref{theoMeasure},
  \ref{theoEpsilonOne}, and \ref{theoEpsilonTwo} from this section.
\end{proof*}
Although Theorem \ref{theoMeasureErgodicAnalogue} does not refer to
Kolmogorov complexity, an open question remains how parameter
$\epsilon_Q$ can be evaluated in nontrivial cases (i.e., for a process
not being a memoryless source).  In Section \ref{secEvaluation}, we
will exhibit two processes for which $\delta^+_Q=\gamma^+_Q$ and
$\delta^-_Q=\gamma^-_Q$. Our evaluation of Hilberg exponents for these
processes is direct, without bounding the parameter $\epsilon_Q$.  We
are not aware of any process for which $\delta^+_Q>\gamma^+_Q$ or
$\delta^-_Q>\gamma^-_Q$.

%\section{Computable universal codes}
%\label{secOther}

\section{Exponents for particular sources}
\label{secEvaluation}

Hilberg exponents can be effectively evaluated in certain cases.  In
this section we shall compute exponents $\gamma^\pm_Q$ and
$\delta^\pm_Q$ related to the underlying measure $Q$ of the process.
For IID processes, these Hilberg exponents are trivially equal zero
since there is no dependence in the process. Equalities
$\delta^\pm_Q=0$ hold also for Markov processes over a finite alphabet
and hidden Markov processes with a finite number of hidden states,
since the expected mutual information is bounded for measures of those
processes by the data-processing inequality. Hence, in that case, we
also have $\gamma^\pm_Q=0$ by Theorem
\ref{theoMeasureErgodicAnalogue}.

Some simple example of a process with unbounded mutual information is
the mixture of Bernoulli processes over the alphabet
$\mathbb{X}=\klam{0,1}$, which we will call the mixture Bernoulli
process:
\begin{align}
  \label{BayesianBernoulli}
  Q(x_1^n)= 
  \int_0^1 \theta^{\sum_{i=1}^n x_i}(1-\theta)^{n-\sum_{i=1}^n x_i}
  d\theta
  =
  \frac{1}{n+1}
  \binom{n}{\sum_{i=1}^n x_i}^{-1}
  .
\end{align}
Although $X_i$ are dependent for this measure $Q$, we will show that
the related Hilberg exponents also vanish.  

It should be noted that the mixture Bernoulli process is a
conditionally IID $1$-parameter source.  The pointwise mutual
information $I^Q(n)$ for conditionally IID sources is equal to a
difference of redundancies. Moreover, as shown in
\cite{Atteson99,BarronRissanenYu98,LiYu00}, for $k$-parameter
processes, the redundancy is proportional to $k\log n$.  Formally,
this suffices to prove that the Hilberg exponents for the mixture
Bernoulli process are zero. Nevertheless, we feel it may be better to
present a complete calculation, which is not that long.  By the
results of \cite{LiYu00}, our reasoning can be generalized to mixtures
of $k$-parameter exponential families but we skip this topic to
present a simple example in a sufficient detail.

For the direct evaluation of the Hilberg exponents, it is convenient
to introduce a few further notations.  Let the (expected) entropy of a
random variable $X$ be written as
\begin{align}
  H_Q(X)=\sred_Q \kwad{-\log Q(X)}
  ,
\end{align}
whereas the (expected) mutual information between variables $X$ and
$Y$ will be written as 
\begin{align}
  I_Q(X;Y)=H_Q(X)+H_Q(Y)-H_Q(X,Y)
  .
\end{align}
Moreover, we define the partial sums
\begin{align}
  T_n&=\sum_{i=-n+1}^0 X_i
  ,
  \\
  S_n&=\sum_{i=1}^n X_i
  .
\end{align}

Now we can state the following result for the expected Hilberg exponents.
\begin{theorem}
  For measure (\ref{BayesianBernoulli}), we have $\delta^+_Q=\delta^-_Q=0$.
\end{theorem}
\begin{proof}
   It can be easily shown that $X_{-n+1}^0$
  and $X_1^n$ are conditionally independent given $T_n$ and
  $S_n$. Hence
  \begin{align}
    I^Q(n)
    =
    -\log \frac{Q(T_n)Q(S_n)}{Q(T_n,S_n)}
  \end{align}
  so the expected mutual information equals $\sred_Q
  I^Q(n)=I_Q(T_n;S_n)$.  Variable $S_n$ assumes under $Q$ each value
  in $\klam{0,1,...,n}$ with equal probability $(n+1)^{-1}$. Hence
  $0\le I_Q(T_n;S_n)\le H_Q(S_n)=\log (n+1)$, which implies the claim.
\end{proof}

The random Hilberg exponents $\gamma^\pm_Q$ for the mixture Bernoulli
process also vanish. This follows from $\delta^\pm_Q=0$ by Theorem
\ref{theoMeasureErgodicAnalogue}.  It may be insightful, however,
to compute $\gamma^\pm_Q$ directly, following the calculation scheme
in \cite{Atteson99,BarronRissanenYu98,LiYu00}.
\begin{theorem}
  For measure (\ref{BayesianBernoulli}), 
  $\gamma^+_Q=\gamma^-_Q=0$ holds $Q$-almost surely.
\end{theorem}
\begin{proof}
  Measure $Q$ defined in (\ref{BayesianBernoulli}) is not ergodic. Its
  random ergodic measure (\ref{RandomErgodic}) takes values of the IID
  measures
\begin{align*}
  E(x_1^n)= 
  \theta^{\sum_{i=1}^n x_i}(1-\theta)^{n-\sum_{i=1}^n x_i}
  ,
\end{align*}
where $\theta$ is a random variable uniformly distributed on $(0,1)$.
Now we will show that $\gamma^+_Q=\gamma^-_Q=0$ holds $E$-almost
surely for any $\theta$, which implies that $\gamma^+_Q=\gamma^-_Q=0$
holds $Q$-almost surely. For this aim we will use the Stirling
approximation
\begin{align*}
  n!=\sqrt{2\pi n}\okra{\frac{n}{e}}^n\okra{1+o(1)}
  .
\end{align*}
Hence the logarithm of the binomial coefficient is
\begin{align*}
  \binom{n}{k}=
  \frac{1}{2}\log\frac{1}{2\pi}
  +
  \frac{1}{2}\log\frac{n}{k(n-k)}+
  nH\okra{\frac{k}{n}}
  + o(1)
  ,
\end{align*}
where $H(p)=-p\log p -(1-p)\log(1-p)$ is the entropy of probability
distribution $(p,1-p)$.  Thus we obtain
\begin{align*}
  I^Q(n)
  &=
  \log\frac{(n+1)^2}{2n+1}
  +
  \log \frac{\binom{n}{T_n}\binom{n}{S_n}}{\binom{2n}{T_n+S_n}}
  \\
  &=
  \log\frac{(n+1)^2}{2n+1}
  +
  \frac{1}{2}\log\frac{1}{2\pi}
  \\
  &\phantom{=}
  +
  \frac{1}{2}\log\frac{n}{T_n(n-T_n)}
  +
  \frac{1}{2}\log\frac{n}{S_n(n-S_n)}
  \\
  &\phantom{=}
  -
  \frac{1}{2}\log\frac{n}{(T_n+S_n)(n-T_n-S_n)}
  \\
  &\phantom{=}
  +
  nH\okra{\frac{T_n}{n}}+nH\okra{\frac{S_n}{n}}-2nH\okra{\frac{T_n+S_n}{2n}}
  \\
  &\phantom{=}
  +
  o(1)
  .
\end{align*}

The sequel is straightforward. Define the partial sums
$T_n=\sum_{i=-n+1}^0 X_i$ and $S_n=\sum_{i=1}^n X_i$.  Quotients
$S_n/n$ and $T_n/n$ converge to $\theta$ $E$-almost surely.  Further,
we may use the Taylor expansion
\begin{align*}
  H\okra{\frac{T_n}{n}}
  =
  H\okra{\theta}
  +
  H'\okra{\theta}\okra{\frac{T_n}{n}-\theta}
  +
  \frac{1}{2}H''\okra{\theta_1}\okra{\frac{T_n}{n}-\theta}^2
  ,
\end{align*}
where $\theta_1\in\kwad{\frac{T_n}{n},\theta}$, and its analogues for
other entropies. This yields
\begin{align*}
  I^Q(n)
  &=
  \log\frac{(n+1)^2}{2n+1}
  +
  \frac{1}{2}\log\frac{1}{2\pi\theta(1-\theta)}
  \nonumber
  \\
  &\phantom{=}
  +
  \frac{1}{2}nH''\okra{\theta_1}\okra{\frac{T_n}{n}-\theta}^2
  +
  \frac{1}{2}nH''\okra{\theta_2}\okra{\frac{S_n}{n}-\theta}^2
  \nonumber
  \\
  &\phantom{=}
  -
  nH''\okra{\theta_3}\okra{\frac{T_n+S_n}{2n}-\theta}^2
  +
  o(1)
  .
\end{align*}
By the law of the iterated logarithm,
\begin{align*}
  \limsup_{n\rightarrow\infty} 
  \frac{\abs{T_n-n\theta}}{\sqrt{n\log\log n}}=C
\end{align*}
for a cerrtain $C\in(0,\infty)$ holds $E$-almost surely, and we have
similar laws for $S_n$ and $T_n+S_n$. Hence
\begin{align*}
  \limsup_{n\rightarrow\infty}
  \abs{
  I^Q(n)
  -
  \log\frac{(n+1)^2}{2n+1}
  -
  \frac{1}{2}\log\frac{1}{2\pi\theta(1-\theta)}
  }
  &\le A\log\log n
\end{align*}
for a certain $A\in(0,\infty)$.  Thus $\gamma^+_Q=\gamma^-_Q=0$ holds
$E$-almost surely for any $\theta$.
\end{proof}

In the next example we will exhibit a process for which Hilberg
exponents do not vanish. This process, introduced in
\cite{Debowski11b,Debowski12} under the name of a Santa Fe process is
also conditionally IID (nonergodic) but does not constitute a
$k$-parameter source. Its construction is partly motivated
linguistically.  Namely, we have certain statements $X_i$ that
describe for randomly selected indices $K_i=k$ the values of some
random binary variables $Z_k$, where the set of available indices
is countably infinite, $k\in\mathbb{N}$.

Formally, the Santa Fe process $(X_i)_{i\in\mathbb{Z}}$ is a sequence
of variables $X_i$ which consist of pairs
\begin{align}
  \label{exUDP}
  X_i&=(K_i,Z_{K_i})
  ,
\end{align}
where processes $(K_i)_{i\in\mathbb{Z}}$ and $(Z_k)_{k\in\mathbb{N}}$
are independent and distributed as follows. First, variables $Z_k$
are binary and uniformly distributed,
\begin{align}
  Q(Z_k=0)=Q(Z_k=1)&=1/2
  ,
  &
  (Z_k)_{k\in\mathbb{N}}&\sim \text{IID}
  .
\end{align}
Second, variables $K_i$ obey the power law
\begin{align}
  \label{ZetaK}
  Q(K_i=k)&=k^{-1/\beta}/\zeta(\beta^{-1})
  , 
  &
  (K_i)_{i\in\mathbb{Z}}&\sim \text{IID}
  ,
\end{align}
where $\beta\in(0,1)$ is a parameter and $\zeta(x)=\sum_{k=1}^\infty
k^{-x}$ is the zeta function. Let us note that, formally, random
variable $Y=\sum_{k=1}^\infty 2^{-k}Z_k$ could be considered a single
random real parameter of the process but the distribution of the
process $(X_i)_{i\in\mathbb{Z}}$ is not a differentiable function of
this parameter.  For this reason the Santa Fe process is not a
$1$-parameter source.

Like in the case of the mixture Bernoulli process, the Hilberg
exponents for the Santa Fe process are all equal but, unlike the case
of the mixture Bernoulli process, they do not vanish. Their common
value is the parameter $\beta$ in the distribution (\ref{ZetaK}).
\begin{theorem}
  \label{theoSantaFeExpected}
  For process (\ref{exUDP}), we have $\delta^+_Q=\delta^-_Q=\beta$.
\end{theorem}
\begin{proof}
By \cite[Proposition 1]{Debowski12}, $\sred_Q I^Q(n)$ grows
proportionally to $n^\beta$. This implies the claim.
\end{proof}
\begin{theorem}
  For process (\ref{exUDP}), $\gamma^+_Q=\gamma^-_Q=\beta$ holds
  $Q$-almost surely.
\end{theorem}
\begin{proof}
  By Theorems~\ref{theoMeasureErgodicAnalogue}
  and~\ref{theoSantaFeExpected} it suffices to prove that
  $\gamma^-_Q\ge\beta$. Let $V(k_1^n)$ denote the set of distinct
  values in sequence $k_1^n$. We have
  \begin{align*}
    Q(X_1^n)=Q(K_1^n)2^{-\card V(K_1^n)}
    .
  \end{align*}
  Hence
  \begin{align*}
    I^Q(n)
    &=\card V(K_{-n+1}^0)+\card V(K_1^n)-\card V(K_{-n+1}^n)
    \\
    &=\card (V(K_{-n+1}^0)\cap V(K_1^n))
    .
  \end{align*}

  Let $L_n=n^{\beta(1-\epsilon)}$, where $\epsilon>0$. We have 
  \begin{align*}
    &Q\okra{\klam{1,2,...,\floor{L_n}}\not\subset V(K_{-n+1}^0)\cap
      V(K_1^n)}
    \\
    &\le
    \sum_{k=1}^{\floor{L_n}} Q(k\not\in V(K_{-n+1}^0)\cap V(K_1^n))
    \\
    &\le
    \sum_{k=1}^{\floor{L_n}} 
    \kwad{
      Q(k\not\in V(K_{-n+1}^0))
      + 
      Q(k\not\in V(K_1^n)) 
    }
    \\
    &=
    \sum_{k=1}^{\floor{L_n}} 2(1-Q(K_i=k))^{n}
    \\
    &\le 
    2L_n \kwad{1-\frac{L_n^{-1/\beta}}{\zeta(\beta^{-1})}}^n
    \\
    &= 
    2L_n \exp\kwad{n\ln (1-L_n^{-1/\beta}/\zeta(\beta^{-1}))}
    \\
    &\le
    2L_n \exp\kwad{-nL_n^{-1/\beta}/\zeta(\beta^{-1})}
    \\
    &\le 
    2n^\beta \exp\kwad{-n^{\epsilon}/\zeta(\beta^{-1})}
    .
  \end{align*}

  Since 
  \begin{align*}
    \sum_{n=1}^\infty
    Q\okra{\klam{1,2,...,\floor{L_n}}\not\subset V(K_{-n+1}^0)\cap
      V(K_1^n)}<\infty
    ,
  \end{align*}
  hence, by the Borel-Cantelli lemma, sets
  $\klam{1,2,...,\floor{L_n}}$ are $Q$-almost surely subsets of
  $V(K_{-n+1}^0)\cap V(K_1^n)$ for all but finitely many $n$. In
  consequence, $I^Q(n)\ge \floor{n^{\beta(1-\epsilon)}}$ for those
  $n$, which implies $\gamma^-_Q\ge\beta$ since $\epsilon$ was chosen
  arbitrarily.
\end{proof}

It should be noted that both the measures of the mixture Bernoulli
process and the Santa Fe processes are nonergodic and
computable. Hence Corollaries \ref{theoMeasureComputable} and
\ref{theoMeasureNonergodic} apply to these sources.  There exist also
mixing (i.e., ergodic in particular) and computable measures $Q$ for
which exponents $\delta^+_Q=\delta^-_Q$ assume an arbitrary value in
$(0,1)$. These processes can be constructed as a modification of the
original Santa Fe process (\ref{exUDP}). For the construction, see
\cite{Debowski12}.
 
The third example will be a process for which the upper and the lower
expected Hilberg exponents are different. The process is a slight
modification of the Santa Fe process, though different than that
discussed in \cite{Debowski12}. Consider a sequence of fixed numbers
$(a_k)_{k\in\mathbb{N}}$ where $a_k\in\klam{0,1}$. Let
\begin{align}
  \label{exUDP2}
  X_i&=(K_i,Y_{K_i})
  ,
\end{align}
where $Y_k=a_kZ_k$, whereas processes $(K_i)_{i\in\mathbb{Z}}$ and
$(Z_k)_{k\in\mathbb{N}}$ are independent and distributed as for the
original Santa Fe process. If $a_k\neq 0$ for some $k$, process
(\ref{exUDP2}) is also nonergodic.
\begin{theorem}
  There exists such a sequence $(a_k)_{k\in\mathbb{N}}$ that for
  process (\ref{exUDP2}), we have $\delta^+_Q=\beta$ and
  $\delta^-_Q=0$.
\end{theorem}
\begin{proof}
Analogously as in the proof of \cite[Proposition 1]{Debowski12}, we
obtain
\begin{align*}
  %\label{EnUDP}
  \sred_Q I^Q(n)=\sum_{k=1}^\infty
  a_k\okra{1-\okra{1-\frac{A}{k^{1/\beta}}}^n}^2
  ,
\end{align*}
where $A:=1/\zeta(\beta^{-1})$. In case of the original Santa Fe
process, we have $a_k=1$ for all $k$ and then $\sred_Q I^Q(n)$ is
asymptotically proportional to $n^\beta$ \cite[Proposition
1]{Debowski12}. Thus, it is sufficient to show that
$\delta^+_Q\ge\beta$ and $\delta^-_Q\le 0$ for a certain sequence
$(a_k)_{k\in\mathbb{N}}$.

We have two bounds
\begin{align*}
  \okra{1-\frac{A}{k^{1/\beta}}}^n&\ge
  \max\klam{0,\frac{nA}{k^{1/\beta}}}
  ,
  \\
  \okra{1-\frac{A}{k^{1/\beta}}}^n&=
  \exp\okra{n\ln\okra{1-\frac{A}{k^{1/\beta}}}}
  \le
  \exp\okra{-\frac{nA}{k^{1/\beta}}}
  .
\end{align*}
Hence
\begin{align*}
  \sred_Q I^Q(n)&\le
  \sum_{k=1}^{\floor{n^{2\beta}}} a_k +
  \sum_{k=\floor{n^{2\beta}}+1}^\infty a_k
  \okra{\frac{nA}{k^{1/\beta}}}^2
  \\
  &\le
  \sum_{k=1}^{\floor{n^{2\beta}}} a_k +
  1+
  A^2n^2
  \int_{n^{2\beta}}^\infty\frac{1}{k^{2/\beta}} dk
  \\
  &=
  \sum_{k=1}^{\floor{n^{2\beta}}} a_k +
  1+
  A^2n^2
  \frac{\beta}{2-\beta} (n^{2\beta})^{\frac{\beta-2}{\beta}}
  \le
  \sum_{k=1}^{\floor{n^{2\beta}}} a_k +
  \okra{1+\frac{A^2\beta}{2-\beta}}
  ,
  \\
  \sred_Q I^Q(n)&\ge
  \sum_{k=1}^\infty a_k
  \okra{1-\exp\okra{-\frac{nA}{k^{1/\beta}}}}
  \ge \sum_{k=1}^{\floor{n^\beta}} a_k
  \okra{1-\exp\okra{-A}}
  .
\end{align*}

Having these two auxiliary results, we can easily show that
$\delta^+_Q\ge\beta$ and $\delta^-_Q\le 0$ for the following sequence
$(a_k)_{k\in\mathbb{N}}$. Let $(b_m)_{m\in\mathbb{N}}$ and
$(c_m)_{m\in\mathbb{N}}$ be two sequences of natural numbers, where
additionally $c_0=0$
and $$\floor{c_{m-1}^{\beta}}<\floor{b_m^{2\beta}}<\floor{c_m^{\beta}}$$
for all $m$, and let $(\epsilon_m)_{m\in\mathbb{N}}$ be a sequence of
real numbers $\epsilon_m=\beta/m$. We put
\begin{align*}
  a_k=
  \begin{cases}
  0, &  \floor{c_{m-1}^{\beta}}<k\le \floor{b_m^{2\beta}},
  \\
  1, & \floor{b_m^{2\beta}}<k\le \floor{c_m^{\beta}}
  .
  \end{cases}
\end{align*}
As for sequences $(b_m)_{m\in\mathbb{N}}$ and
$(c_m)_{m\in\mathbb{N}}$, we choose them to satisfy
\begin{align*}
  \floor{c_{m-1}^{\beta}}+\okra{1+\frac{A^2\beta}{2-\beta}}
  &\le b_m^{\epsilon_m}
  ,
  \\
  (\floor{c_m^{\beta}}-\floor{b_m^{2\beta}})\okra{1-\exp\okra{-A}}
  &\ge
  c_m^{\beta-\epsilon_m}
  .
\end{align*}
In this way we obtain
\begin{align*}
  \sred_Q I^Q(b_m)&\le
  \floor{c_{m-1}^{\beta}}+\okra{1+\frac{A^2\beta}{2-\beta}}
  \le b_m^{\epsilon_m}
  ,
  \\
  \sred_Q I^Q(c_m)&\ge
  \okra{\floor{c_m^{\beta}}-\floor{b_m^{2\beta}}}\okra{1-\exp\okra{-A}}
  \ge
  c_m^{\beta-\epsilon_m}
  .
\end{align*}
Hence $\delta^+_Q\ge\beta$ and $\delta^-_Q\le 0$, as requested.
\end{proof}

Evaluation of random Hilberg exponents for process (\ref{exUDP2})
seems difficult.

\section*{Acknowledgment}

The author thanks the anonymous referees for helpful comments that
encouraged him to improve the composition of this paper.

\bibliographystyle{IEEEtran}

\bibliography{0-journals-abbrv,0-publishers-abbrv,ai,mine,tcs,ql,books}

\end{document}